\documentclass[conference]{IEEEtran}
\usepackage{ltexpprt}
\usepackage{graphicx}
\usepackage{amsmath}
\usepackage{dsfont}
\usepackage{algorithm}
\usepackage{algorithmic}
\usepackage{array}

\newcommand{\remove}[1]{}
\newcommand{\ac}{\mathrm{ac}}
\newcommand{\cent}{\mathrm{cr}}
\newcommand{\pr}{\mathrm{pr}}
\newcommand{\lacr}{^{la}\mathrm{ac}}
\newcommand{\lapr}{^{la}\mathrm{pr}}

\begin{document}
%
\title{Limited Attention and Centrality in Social Networks}
\author{\IEEEauthorblockN{Kristina Lerman$^1$, Prachi Jain$^2$, Rumi Ghosh$^3$, Jeon-Hyung Kang$^1$ and Ponnurangam Kumaraguru$^2$}
\IEEEauthorblockA{1. USC Information Sciences Institute, Marina del Rey, CA 90292 USA\\
2. Indraprastha Institute of Information Technology Delhi, India\\
3. HP Labs, Palo Alto, CA, USA }
}

\remove{
\author{\IEEEauthorblockN{Kristina Lerman \\ and Jeon-hyung Kang}
\IEEEauthorblockA{USC Information Sciences Institute\\
Marina del Rey, CA 90292\\
Email: \{lerman@isi.edu,jeonh@usc.edu\}}
\and
\IEEEauthorblockN{Prachi Jain \\ and Ponnurangam Kumaraguru}
\IEEEauthorblockA{Indraprastha Institute of Information Technology\\}
\and
\IEEEauthorblockN{Rumi Ghosh}
\IEEEauthorblockA{HP Labs\\}
}
}

\maketitle

\begin{abstract} 
How does one find important or influential people in an online social network? Researchers have proposed a variety of centrality measures to identify individuals that are, for example, often visited by a random walk, infected in an epidemic, or receive many messages from friends. Recent research suggests that a social media users' capacity to respond to an incoming message is constrained by their finite attention, which they divide over all incoming information, i.e., information sent by users they follow. We propose a new measure of centrality --- limited-attention version of Bonacich's Alpha-centrality --- that models the effect of limited attention on epidemic diffusion. The new measure describes a process in which nodes broadcast messages to their out-neighbors, but the neighbors' ability to receive the message depends on the number of in-neighbors they have.
We evaluate the proposed measure on  real-world online social networks and show that it can better reproduce an empirical influence ranking of users than other popular centrality measures.
\end{abstract}

\section{Introduction}




An individual's position within a social network is thought to confer advantages, allowing him to exploit the structure of social ties to accumulate power, prestige or
influence~\cite{Bavelas,Granovetter73,Freeman79,Krackhardt93,Burt95,Burt04}.
Many measures of centrality were proposed to capture the  importance of the position in a network. Some of these, like degree and betweenness centrality~\cite{Freeman79}, measure an individual's ability to control the flow of information in the network. Other measures give higher centrality to those positions that are themselves connected to central positions~\cite{Katz53,Bonacich87,PageRank,Bonacich01}.
The growing popularity of online social media has sparked new interest in centrality. Researchers have proposed using centrality to identify influential social media users~\cite{Cha10icwsm,Bakshy11wsdm} whose endorsement can, for example, maximize the reach of a ``viral'' marketing campaign~\cite{Kempe03}, or conversely, who can most quickly stop a malicious rumor from spreading.

Most of the existing centrality measures examine link structure of the network to identify key nodes within it. Take, for example, the Web, which is represented as a directed graph of hyperlinked Web pages. An important page within this graph is one that is visited often by Web surfers. This observation forms the basis of Google's original Web page ranking algorithm PageRank~\cite{PageRank}. By modeling Web surfing as a random walk, PageRank assigns a centrality score to each page based on its value in the equilibrium distribution of the random walk. However,  a central individual in a social network through which disease is spreading is one who infects, either directly or indirectly, most others. Unlike Web surfing, the spread of a virus is modeled as an epidemic process. Thus, PageRank, which is intimately connected with random walks, will not identify key individuals in a social network. Instead, a measure such as the Katz score~\cite{Katz53} or Bonacich's Alpha centrality~\cite{Bonacich87}, which gives the equilibrium distribution of an epidemic process on a network~\cite{Ghosh12nonconservative}, is more appropriate.

Now consider information spreading through an online social network, for instance, by users sending messages or product recommendations to their friends. While information spread in networks is often modeled as an epidemic process (e.g., \cite{Gruhl04,Leskovec07}), recent research suggests that psychological and cognitive factors are important in determining whether a person will \emph{see} and \emph{act} on friends' recommendations. Specifically, attention  was shown to be a critical aspect of online behavior~\cite{Goldhaber97,Wu07,Weng:2012dd,Hodas12socialcom}. Attention is the psychological mechanism that controls how we process incoming stimuli and decide what activities to engage in~\cite{Kahneman73,Rensink:1997vj}. Actions, such as reading a tweet, browsing a Web page, or responding to email, require mental effort, and since human brain's capacity for mental effort is limited, so is attention.  Moreover, online users must divide their attention over all incoming stimuli~\cite{Hodas12socialcom}. As a consequence, the more stimuli people have to process, the smaller the probability they will respond to any one stimulus. While attention need not be distributed uniformly over friends --- some friends may receive a greater share of a person's attention due to familiarity, trust, social closeness, or influence~\cite{Gilbert09,Huberman-attention} --- for simplicity, we assume that each friend receives the same fraction of a person's attention. We call this phenomenon \emph{limited attention} (\emph{la}).

Limited, divided attention changes the nature of interactions between nodes in a network and therefore, how central nodes are identified. Now a node's capacity to infect others depends not only on how many connections it has but also on who and how many others these nodes are connected to. In Section~\ref{sec:centrality}, we introduce a new centrality measure --- limited-attention Alpha-Centrality ($laAC$) --- that models attention-limited nature of social interactions and provide its mathematical definition. 
For completeness, we also introduce and define limited-attention PageRank ($laPR$), which models the effect of limited attention on a random walk process. In Section~\ref{sec:experiments}, we evaluate the proposed algorithms and centrality measures on real-world data, including follower graphs from social media sites Digg and Twitter. In the Appendix, we present fast approximate algorithms that allow us to calculate these measures even on large graphs and provide their performance guarantees.

\section{Dynamics, Attention and Centrality}
Centrality measures examine topology of a network to identify important or central nodes within it.  It has been recognized recently, however, that centrality is the product of a network's \emph{links} and the \emph{dynamical processes} taking place on it, which determine how ideas, pathogens, or influence flow along social links~\cite{Borgatti05,Lambiotte11,Ghosh11nonconservative,Ghosh12nonconservative}. Take, for example, one definition of centrality used by the popular PageRank algorithm~\cite{PageRank}: a network node is important if it is often visited by a random walk.
A random walk is a stochastic process that starts at some node, and at each time step transitions to a randomly selected neighbor of the current node. Variants of the random walk are used to model flows in physical systems, e.g., chemical and heat diffusion, and can be used to model social phenomena resulting from one-to-one interactions, such as Web surfing, money exchange and phone conversations.

\begin{figure}[htbp] %
   \centering
   \begin{tabular}{cc}
   \includegraphics[height=1in]{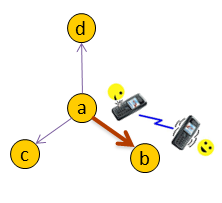} &
   \includegraphics[height=1in]{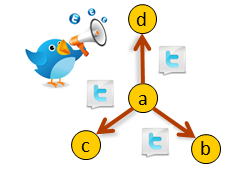}\\
   (a) & (b) \\
   \includegraphics[height=1in]{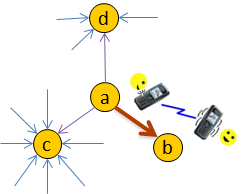} &
   \includegraphics[height=1in]{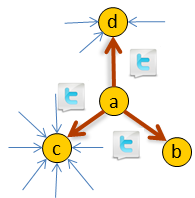}\\
   (c) & (d)
   \end{tabular}
   \caption{Different dynamical processes taking place on a network: (a) random walk, (b) epidemic spread, and their limited-attention variants: (c) limited-attention random walk and (d) limited-attention epidemic spread. In limited-attention process, a node's capacity to receive a message depends on its in-degree.}
   \label{fig:toy}
\vspace{-0.865em}
\end{figure}

In a social network, a message or a virus propagates by being broadcast by an infected individual to \emph{all} her (out-) neighbors. Such processes are modeled as an epidemic (or a contact) process. The difference between it and the random walk is illustrated in Figure~\ref{fig:toy}, which shows the neighborhood of node $a$. Directed edges in this network represent, for example, hyperlinks between Web pages, or who can call whom in a social network, or in the context of social media, they can also indicate that $b$, $c$ and $d$ follow $a$ and receive broadcasts from her. Figure~\ref{fig:toy}(a) illustrates a one-to-one interaction, e.g., phone call, while Fig.~\ref{fig:toy}(b) shows a one-to-many broadcast.

Until now, we have assumed that nodes have an unlimited capacity to receive incoming signals, whether Web surfers, phone calls, or messages from friends. This may not always be the case. Suppose a Web server can receive a limited number of connections, in extreme case only one. Then the probability that a Web surfer starting at $a$ will reach $b$ depends on whether the Web server in charge of  $b$ is able to receive an incoming  request. In a social network, cognitive and perceptual factors can limit a person's capacity to process incoming messages~\cite{Hodas12socialcom}. Such factors collectively figure into the phenomenon we refer to as \emph{limited attention}. This means that the probability a user will respond to a message from a friend decreases with the number of friends she follows. This is illustrated graphically in Fig.~\ref{fig:toy}(c) and (d). Node $b$ is more likely to receive a message from $a$ than node $c$ because $c$ is receiving messages from eight nodes, while $b$ from only one node.

Different dynamic processes lead to different notions of centrality. PageRank is used to find nodes that are often visited by a random walk (with random restarts), while Alpha- (or Bonacich) Centrality identifies nodes that are often infected during an epidemic~\cite{Ghosh12nonconservative}.
Below we define limited-attention PageRank and limited-attention Alpha-Centrality,  centrality measures that take into account the finite attention of online social users. Limited-attention PageRank identifies nodes that are often visited by a random walk, when each node's capacity to receive the walker depends on its in-degree. Similarly, limited-attention Alpha-Centrality identifies nodes that are often infected in an epidemic, when each node's susceptibility to infection also depends on its in-degree.

\section{Limited-Attention Centrality}
\label{sec:centrality}
We represent a network as a directed  graph  with $V$ nodes and $E$ edges. The adjacency matrix of the graph is defined as: $A[u,v]= 1$ if there is an edge from $u$ to $v$; otherwise, $A[u,v]= 0$. Also, $A[u,u]=0$. The set of  out-neighbors of $u$ is $\lbrace v \in V \vert (u,v) \in E \rbrace$; and the set of in-neighbors is $\lbrace v \in V \vert (v,u) \in E \rbrace$. Two other important quantities are the in-degree and out-degree matrices.  The out-degree matrix  $D_{out}$ is a diagonal matrix defined as $D_{out}[i,i]=\sum_j A[i,j]=Ae^T$ and $D_{out}[i,j]=0$ $\forall$ $i\neq j$. Here, $e$ is a $|V|$-dimensional row vector of ones, and $e^T$ is its transpose. The in-degree matrix  $D_{in}$ is a diagonal matrix defined as $D_{in}[i,i]=\sum_i A[i,j]=eA$ and $D_{in}[i,j]=0$ $\forall$ $i\neq j$.

\subsection{Limited-attention PageRank}
\remove{
\begin{figure*}[htbp] %
   \centering
   \begin{tabular}{cc}
   \includegraphics[height=1.5in]{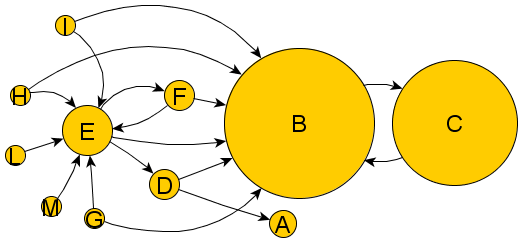} &
   \includegraphics[height=1.5in]{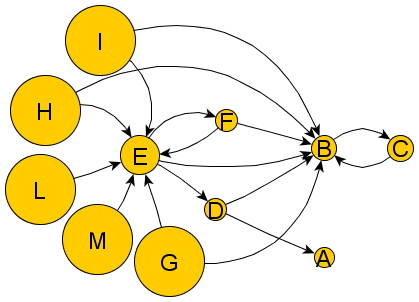}\\
   (a) PageRank & (b) limited-attention PageRank\\
   \end{tabular}
   \caption{Directed network with sizes of nodes weighed by their score according to (a) PageRank and (b) attention-limited PageRank.}
   \label{fig:PR}
\end{figure*}
}

A PageRank vector ${\pr}(\alpha, s)$ is the steady state probability distribution of a random walk with restarts with a damping factor $\alpha$. This means that with a probability $\alpha$, the walk transitions to one of the out-neighbors of a current node, and with probability ($1-\alpha$) it transitions to any node in the network.  The starting vector $s$, gives the probability distribution for where the walk transitions after restarting, which is usually taken as a uniform vector $s=e/|V|$.   The transfer matrix $D_{out}^{-1}A$ encodes the transition probabilities of a random walk on the network. PageRank vector ${\pr}(\alpha,s)$ is the unique solution of the following iterative equation:
\begin{equation}
\label{eq:pr}
{\pr}(\alpha, s) = (1-\alpha) s +  \alpha  {\pr}(\alpha, s)D_{out}^{-1}A
\end{equation}

Now, if a node's capacity to receive a random walker is limited, the transfer matrix must be modified. As stated above, we consider the simplest scenario in which the finite capacity is divided uniformly between all incoming connections. This case is modeled by the transfer matrix $D_{out}^{-1}AD_{in}^{-1}$. Therefore, limited-attention PageRank $\lapr(\alpha,s)$ is the solution of the following iterative equation:
\begin{equation}
\label{eq:lapr}
{\lapr}(\alpha, s) = (1-\alpha) s +  \alpha {\lapr}(\alpha, s)D_{out}^{-1}AD_{in}^{-1}
\end{equation}
\noindent The starting vector above is $s=e D_{in}^{-1}$. Note that while the PageRank transfer matrix $D_{out}^{-1}A$ is stochastic, since each row
or column sums to one, this is no longer the case for the limited-attention PageRank transfer matrix.

\begin{figure*}[htbp] %
   \centering
   \begin{tabular}{cccc}
   \includegraphics[height=0.9in,width=42mm]{fig/acmGraph_PR_0_85} &
   \includegraphics[height=0.9in]{fig/acmGraph_laPR_0_85}&
   \includegraphics[height=0.9in]{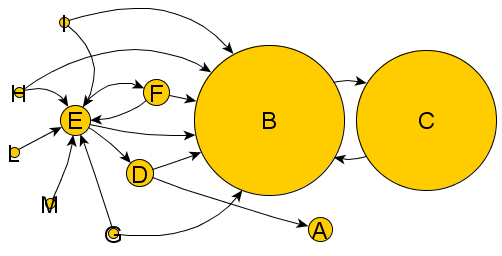} &
   \includegraphics[height=0.9in]{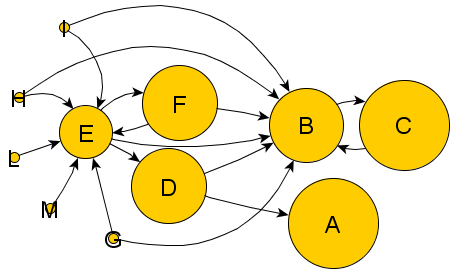}\\
   (a) PR & (b) laPR &(c) AC & (d) laAC \\
   \end{tabular}
   \caption{Directed network with sizes of nodes weighed by their score according to (a) PageRank and (b) attention-limited PageRank (c) Alpha-centrality and (d) limited-attention Alpha-centrality of the influence graph.}
   \label{fig:PR}
\end{figure*}

We illustrate the differences between PageRank and limited-attention PageRank on a toy directed network. Figure~\ref{fig:PR}(a) shows this network with the size of the node proportional to its centrality score relative to other nodes, as determined by PageRank (with $\alpha=0.85$). Node $B$ is the most central, since it has many in-links, enabling a random walker to reach it via many different paths. Peripheral nodes $H$, $I$, $J$, etc., are less important, since they only receive the random walker via a random jump. On the other hand, limited-attention PageRank,  shown in Fig.~\ref{fig:PR}(b), scores these nodes highly. The node ranked highest by PageRank, $B$, on the other hand, dramatically decreases in centrality. This node divides its attention among many in-links, limiting its ability to receive a random walker along any specific link. The peripheral nodes, on the other hand, have few in-links, and are better able to receive the random walker, whether it is following an out-link or executing a random jump. Their importance, therefore, is greater in this scenario.


\subsection{Limited-attention Alpha-Centrality}

Alpha-Centrality measures the total number of paths from a node, exponentially attenuated by their length. Bonacich introduced this measure~\cite{Bonacich87} as a generalization of the index of status proposed by Katz~\cite{Katz53}, and it is sometimes referred to as Bonacich centrality.
Alpha-Centrality matrix gives the number of attenuated paths between two nodes, and it is usually written as a power series expansion of the adjacency matrix, with attenuation parameter $\alpha \ge 0$:
$C=A + \alpha A^2 + \alpha^2 A^3 + \alpha^3 A^4 + \ldots$. 
This series converges to $C=\alpha A (I-\alpha A)^{-1}$ while $\alpha < {1}/{\lambda_{max}}$, where $\lambda_{max}$ is the largest eigenvalue of $A$ (i.e., spectral radius of the network). Parameter $\alpha$  determines how far, on average, a node's effect will be felt and sets the length scale of interactions. When $\alpha$ is small, Alpha-Centrality probes only the local structure of the network. As $\alpha$ grows, more distant nodes contribute to the centrality score of a given node~\cite{Ghosh11physrev}. As $\alpha \to 1/{\lambda_{max}}$,  the length scale of interactions diverges and it becomes a global measure.

Alpha-Centrality gives the steady state distribution of an epidemic process on a network~\cite{Ghosh12nonconservative}, where $\alpha$ is the probability to transmit a message or influence along a link. Therefore, $(i,j)$th entry of the Alpha-Centrality matrix $C$ can be interpreted as the likelihood that the  virus will reach node $j$ from node $i$. Summing over all columns $j$ gives the Alpha-Centrality score of node $i$, $\ac(\alpha)=C e^T=\sum_j C(i,j)$,  or the number of infections directly or indirectly caused by node $i$. Summing over the rows of the Alpha-Centrality matrix, on the other hand, gives $\ac(\alpha)^T=e \cdot C=\sum_i C(i,j)$,  the total number of times that node $i$ is infected by others.

Alpha-Centrality vector $\ac({\alpha},s)$ can also be defined iteratively as:
\begin{equation}
\ac({\alpha},s)=s+ \alpha A \cdot \ac({\alpha},s),
\label{a-cen1}
\end{equation}
\noindent where the starting vector $s=Ae^T$ is taken as out-degree centrality~\cite{Bonacich87}.

Let us now consider the case in which a node's capacity to receive incoming stimuli --- whether messages or viruses --- is limited and uniformly divided among all incoming connections. Therefore, the probability that node $j$ will receive a message broadcast by $i$ will be proportional to $1/d_{in}(j)$, where $d_{in}(j)$ is the in-degree of node $j$. The limited-attention Alpha-Centrality matrix can be written in terms of the modified adjacency matrix $M=AD_{in}^{-1}$ as:
\begin{equation}
\label{eq:lacr}
C_{la}=M + \alpha M^2 + \alpha^2 M^3 + \alpha^3 M^4 + \ldots \nonumber
\end{equation}
The limited-attention Alpha-Centrality vector $\lacr({\alpha},s)$ can also be written in iterative form:
\begin{equation}
\lacr({\alpha},s)=s+ \alpha AD_{in}^{-1} \cdot \lacr({\alpha},s),
\label{eq:laac}
\end{equation}
\noindent with the starting vector $s=AD_{in}^{-1}e^T$.  Note that the transfer matrix $AD_{in}^{-1}$ is a stochastic matrix.

Figures~\ref{fig:PR}(c) and (d) illustrate the differences between Alpha-Centrality and its limited-attention variant. Figure~\ref{fig:PR}(c) shows the directed network with nodes sizes proportional to their $\ac$ scores.  The Alpha-Centrality scores in this example were calculated for $\alpha = 0.85$. The rankings of nodes are similar to those produced by PageRank (Fig.~\ref{fig:PR}(a)), though node $E$, for example, is relatively less important. In the limited-attention variant, shown in Fig.~\ref{fig:PR}(d),  the picture looks completely different. While $B$ in (d) loses its importance, due to may in-links, node $A$ becomes more central, since it receives incoming signals over a single in-link. Peripheral nodes are not judged to be central, because, unlike random jumps in PageRank, they never receive any signals.

\section{Applications to Social Media}
\label{sec:experiments}
We use centrality measures proposed in this paper to identify influential people on social media. Correctly identifying such people can have far-reaching consequences for identifying noteworthy content, targeted information diffusion, and other applications. While calculating Eq.~\ref{eq:lacr} was infeasible for such large networks, we used approximate algorithms presented in the Appendix for these calculation. Appendix also gives performance guarantees of the approximate algorithms.

Researchers have proposed a number of simple heuristics to identify influential social media users that rely, for example, on the number of followers or mentions~\cite{Cha10icwsm,Lee10www,Bakshy11wsdm}. Others have used centrality by analyzing the follower graph to find users with high PageRank scores~\cite{twitter,Tang12}. However, since information spread on networks is traditionally described as an epidemic~\cite{Gruhl04,Leskovec07}, Alpha-Centrality may do a better job~\cite{Ghosh10snakdd}, since it explicitly models epidemic dynamics. We show, however, that limited-attention Alpha-Centrality, the measure that  accounts for both the epidemic nature of social media broadcasts and the divided attention of its users, does a better job identifying influential users than Alpha-Centrality.

Specifically, we study URL-sharing activity on  Digg and Twitter, two popular social media sites for content sharing.
Both sites allow users to follow other users by listing them as friends. The follower relation is asymmetric. When user $A$ follows (becomes as fan of) $B$, she receives $B$'s broadcasts, but not vice versa: we denote the relationship as $B \to A$. Representing the follower graph in matrix form, a user's out-degree measures the number of followers she has, and her in-degree the number of friends she follows.

\subsection{Data Collection}
The Digg dataset\remove{http://www.isi.edu/$\sim$lerman/downloads/digg2009.html} contains more than 3 million votes on some 3500 stories promoted to Digg's front page in June 2009. More than 139K distinct users voted for at least one story in the data set (submission counts as the story's first vote). We call these users \emph{active} users. Next, we extracted the friendship links created by {active} users and constructed a follower graph that contained active users who were following the activities of others.
Only about 71K active users listed others as friends, resulting in network with around 280K users and over 1.7 million links.

The Twitter data set was collected over a period of three weeks in October 2010 using the Gardenhose streaming API. We focused on tweets that included a URL in the body of the message. In order to ensure that we had the complete tweeting history of the URL, we used the search API to retrieve all tweets containing that URL. Users who tweeted the URL are considered \emph{active}. Data collection process resulted in more than 3 million posts tweeted by 816K users which mentioned 70K distinct  URLs. Next, we used the REST API to collect followers of  each active user, keeping only those followers who themselves were active, i.e., tweeted at least one URL during data collection period. The resulting follower graph had almost 700K nodes and over 36 million edges. More details of the data collection method are provided in \cite{Ghosh12nonconservative}.

\subsection{Results}

We calculate Alpha-Centrality ($AC$) and limited-attention Alpha-Centrality ($laAC$) on the Digg and Twitter follower graphs using algorithm for $laAC$ (Alg.~\ref{alg:laac})  presented in the Appendix and the algorithm for $AC$ presented in \cite{Ghosh11nonconservative}. These are approximate algorithms with proven performance guarantees. We calculate limited-attention PageRank ($laPR$)  on the transpose of the follower graph using Alg.~\ref{alg:lapr}, since node's influence is related to the number of walks  it generates, rather than receives. The in- and out-degrees were conditioned by adding a small number (0.01) to avoid division to zero.

\begin{figure}[htbp] %
   \centering
   \begin{tabular}{c}
   \includegraphics[width=0.9\columnwidth]{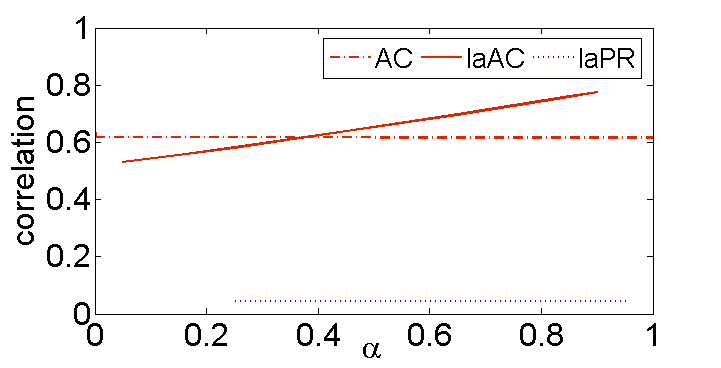} \\
(a) Digg \\
   \includegraphics[width=0.9\columnwidth]{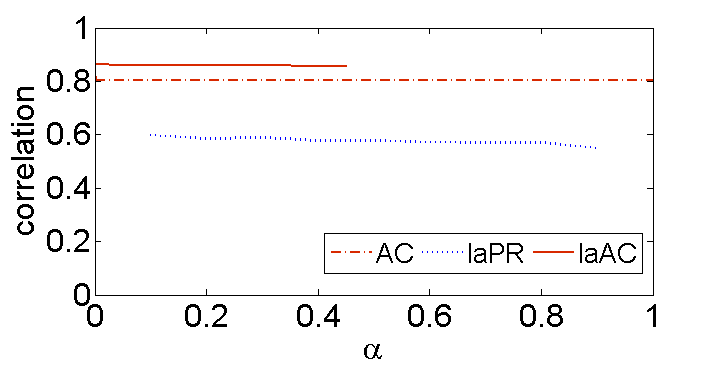} \\
(b) Twitter
   \end{tabular}
   \caption{Correlation of rankings of (a) Digg and (b) Twitter users found by different measures of centrality with the empirical influence ranking.}
   \label{fig:correlations}
\vspace{-0.88em}
\end{figure}

In order to compare the performance of centrality measures, we need a relevant measure of influence. 
When a user posts a URL on Digg or Twitter, she broadcasts it to all her followers. We refer to this user as the \emph{submitter}. Whether or not her follower will re-broadcast the URL (i.e., retweet it on Twitter or vote for it on Digg)  depends on its \emph{quality} and  \emph{submitter's influence}. Assuming that URL's quality is uncorrelated with the submitter, we can average out its effect by aggregating over all URLs submitted by the same user~\cite{Ghosh10snakdd}. The residual difference between submitters can be attributed to variations in influence.
Similar to \cite{Cha10icwsm,Ghosh12nonconservative,Bakshy11wsdm}, we use the average number of times the URLs submitted by the user are re-broadcast by her followers as the \emph{empirical measure of influence}.

Figure~\ref{fig:correlations} shows how well the rankings produced by different centralities correlate with the empirical influence rankings of users who submitted at least two URLs which were rebroadcast at least ten times. We use Spearman rank correlation because it is less sensitive to variations in scores, and we expect some variation to arise in approximate centrality scores.
Limited-attention Alpha-Centrality correlates better with the empirical measure of influence than Alpha-Centrality over a broad range of $\alpha$ values, consistent with our claim that $laAC$ is a better measure for predicting central social media users, because it better models the dynamics of online communication than $AC$. On Digg, $AC$ appears to outperform $laAC$ for small values of  $\alpha$. Since $\alpha$ can be thought of as the scale of interaction, this implies that locally, $AC$ better predicts influential users. This could be the consequence of the fact that our measure of influence, i.e., number of re-broadcasts by followers, is a local measure. In the future, we plan to compare the performance of centrality measures using a global measure of influence, for example, the average size of cascades triggered by submitted URLs.
We did not expect limited-attention PageRank ($laPR$) to predict influence rankings of Digg and Twitter users, since the dynamic process this centrality models does not at all describe communication patterns of social media users, and we found no correlation.

Interestingly, PageRank and $laAC$ have similar performance, since $laAC$ calculated on the adjacency matrix $A$ of the follower graph is almost identical to $PR$ calculated on the transpose of $A$, except that the starting vectors are different in the two algorithms. This suggests that dynamics of random walk are almost equivalent to epidemic dynamics under the conditions of uniformly divided attention, when direction of the flow is reversed. This observation could explain why $PR$ can give good results in the social media domain. We leave implications of this observation for future research.

\section{Conclusion}
Information flow in social networks, including online networks, is often modeled as an epidemic process, suggesting that centrality measures based on epidemics are appropriate for predicting influential social media users. We propose a new centrality measure that takes into account the finite capacity of social media users to process incoming messages from friends. We modeled such limited attention by scaling the probability a node receives a message by the inverse of its in-degree. We presented approximate algorithm that allows us to efficiently calculate proposed measure for the real-world social networks on Digg and Twitter. We showed empirically that centrality measure that models limited-attention epidemics does a better job predicting highly retweeted social media users than one that models simple epidemics. Our findings suggest that the nature of interactions among network nodes should determine how central nodes are identified.

\subsection*{Acknowledgements}
This material is based upon work supported by the Air Force Office of Scientific Research under contracts FA9550-10-1-0569  and FA9550-10-1-0102, by the Air Force Research Laboratories under contract FA8750-12-2-0186, by DARPA under contract W911NF-12-1-0034, and by the National Science Foundation under grant CIF-1217605.
PJ's internship was sponsored by the USC Viterbi-India Summer program.

\appendix
\section*{Appendix: Approximate Algorithms}
\label{sec:algorithms}
Finding limited-attention PageRank (Eq.~\ref{eq:lapr}) and Alpha-Centrality (Eq.~\ref{eq:lacr}) requires the computation of matrix inverse, which can be done in $O(|V|^{3})$ operations using the naive implementation of the algorithm ($|V|$ is the number of nodes in the network). This is prohibitively expensive for networks with thousands  or more nodes.
However, solving  equations iteratively requires $O(|V|^{2})$ operations in each iteration, though we do not know how many iterations are sufficient for an optimal solution. We propose Approximate Limited-Attention Page Rank and Approximate Limited-Attention Alpha Centrality algorithms, which can be used to calculate a near optimal solution. The algorithms use a single error tolerance parameter $\delta $ ($0 < \delta \le 1$) to control both the quality of the solution and computation time.

The proposed algorithms and their performance guarantee are based on the approximate PageRank~\cite{approximatePR} and approximate Alpha-Centrality~\cite{Ghosh11nonconservative} algorithms. They provide a flexible way to compute the near optimal centrality vector $\tilde{\cent}$ using a starting vector $s$ and a residual vector $r$. Initially $r = s$ and $\tilde{\cent}=\vec{0}$. The algorithms iteratively move the weight from ${r}$ to  $ {\tilde{\cent}}$ vector, until the values in the residual vector $r$ are sufficiently small. The amount of error in the approximate centrality vector is equivalent to the amount remaining in the residual vector. The performance guarantee of the proposed algorithms are given in  Theorem~\ref{theorem1} and Theorem~\ref{theorem2} , which are based on Lemma~\ref{lemma1}. The Lemma states that each iteration maintains an invariant vector $\tilde{\cent} = \cent(s) - \cent(r) = \cent(s - r)$. This means that the amount of error in the approximate centrality vector is equivalent to the error remaining in the residual vector.
\begin{proposition}
For any fixed value of $\alpha$ in $[0,1]$ and starting vector $s$, $\cent(\alpha,s)$ is linear in  $s$.
\label{proposition1}
\end{proposition}
\begin{proof}
The limited-attention PageRank vector $ {\lapr}(\alpha, s) $ is a unique solution to
\begin{equation*}
\label{eq:lapr2}
\cent(s)={\lapr}(\alpha, s) = (1-\alpha) s +  \alpha \cdot {\lapr}(\alpha, s)M
\end{equation*}
\noindent where M=$D_{out}^{-1}AD_{in}^{-1}$. The limited-attention Alpha-Centrality vector $\lacr({\alpha},s)$ can also be written in iterative form:
\begin{equation*}
\cent(s)=\lacr({\alpha},s)=s+ \alpha \cdot \lacr({\alpha},s)M,
\label{a-cen2}
\end{equation*}
\noindent where M=$AD_{in}^{-1}$. The centrality vectors can be proved linear with respect to $s$ by substituting suitable values for $\cent(s)$ and  $M$ in the proof presented in \cite{Ghosh11nonconservative}.
\end{proof}

\begin{lemma}
At the start of each iteration of while loop $\tilde{\cent}$ = $\cent(s) - \cent(r)$ = $\cent(s - r)$ such as the sum of elements in $r$ decreases with each iteration.
\label{lemma1}
\end{lemma}
\begin{proof}
The proof of correctness is based on Proposition ~\ref{proposition1}. During initialization, $\emph {r} = {s}$ and $\tilde{\cent}$ = $\vec{0}$; therefore, $\cent(s - r) = \cent(\vec{0}) = \vec{0} = \tilde{\cent}$. The lemma is maintained throughout the execution of the loop. To prove this, we use a row vector $z_u$ such as $z_u(i)=1$ if $i=u$; otherwise, $z_u(i)=0$. Before the next iteration of while loop in Algorithm ~\ref{alg:lapr} we have $\tilde{\cent} = \tilde{\cent} + (1 - \alpha) z_i r(i)$ and $r^{\prime} = r -  z_i r(i) + \alpha r(i)  M z'$ where $\tilde{\cent}^{\prime}, r^{\prime}$ are updated centrality vector and residual vectors and $i$ is the vertex dequeued in line number 11 of the algorithm. Now consider
\begin{eqnarray*}
{\cent}(r) & = & {\cent}(r - z_i r(i)) + {\cent}(z_i r(i))\\
& =&  {\cent}(r - z_i r(i)) + (1 - \alpha) z_i r(i) + {\cent}(\alpha z_i r(i) M)\\
& =&  {\cent}(r - z_i r(i) + \alpha z_i r(i) M) + (1 - \alpha) z_i r(i)\\
& =& {\cent}(r^{\prime}) + \tilde{\cent}^{\prime} - \tilde{\cent}
 ={\cent}(r^{\prime}) + \tilde{\cent}^{\prime} - \cent(s - r)
\end{eqnarray*}

\noindent It follows that $\tilde{\cent}^{\prime} = \cent(r) - \cent(r^{\prime}) + \cent(s - r) = \cent(r - r^{\prime} + (s - r)) = \cent(s - r^{\prime})$. On termination of the loop, given the lemma and an error tolerance parameter the approximate centrality vector should always satisfy
\begin{equation*}
\cent(s)[i]   \ge   \tilde{\cent}[i]   \ge   (1 - \delta) \cent(s)[i]\ \forall i \in V
\end{equation*}
We choose a uniform starting vector $s$,  $s[i]=||s||_{1}/|V|$, $\forall i \in V$.
The algorithm terminates when $r[i]  \le \epsilon d^{max}_{out}$; $\forall i \in V$, so we choose $\epsilon = \frac{\delta ||s||_{1}}{|V| d^{max}_{out}} = \frac{\delta s[i]}{d^{max}_{out}}$.
With this choice of $\epsilon$ we also ensure freedom in choice of the value of $\alpha$ with in the range of 0 to 1. This freedom is achieved at the cost of increased running time of the algorithm.
In the end $r[i] \le \delta s[i]$, therefore, $\implies \cent(r)[i] \le \delta \cent(s)[i]$.
Thus, $$\tilde{\cent}[i] \ge (1 - \delta) \cent(s)[i].$$
It is obvious that $\cent(s)[i] \ge \tilde{\cent}[i]$; hence $\cent(s)[i] \ge \tilde{\cent}[i] \ge (1 - \delta) \cent(s)[i]$  $\forall i \in V$.
Also the sum of all elements of residual vector $\sum{r^{\prime}}$ is
\begin{equation*}
\sum{r^{\prime}}= \sum{r} - r[i] + \left( \alpha \frac{r[i]}{d_{out}(i)} \cdot \sum_{j \in N^{out}(i)}\frac {1} {d_{in}(j)} \right)
\end{equation*}
\noindent Since value of $\alpha$ lies in  [0,1] and $\sum_{j \in N^{out}(i)}\frac {1} {d_{in}(j)} \le d_{out}(i)$, net sum of all values of residual vector decreases with each iteration of while loop.
Similarly the we can prove that the lemma is valid for Algorithm ~\ref{alg:laac}.

\end{proof}

\subsection{Approximate Limited-Attention PageRank}
Limited attention Page Rank ($laPR$) given by Eq.~\ref{eq:lapr}, can be written as the solution $\cent({\alpha,s})$ of:
$$
\cent(\alpha,s)[j]=(1-\alpha)s[j] +\alpha \sum_{i \in N^{in}(j)} \frac{\cent(\alpha,s)[i]}{d_{out}(i) d_{in}(j)}.
$$
\noindent Here $N^{in}(j)$ is a set of in-neighbors of $j$, i.e., nodes $i$ such that edge $(i,j) \in E$. Also, $N^{out}(j)$ is the set of out-neighbors of $j$, i.e., nodes $i$ such that $(j, i) \in E$. 
We take the starting vector $s=e/|V|$ to be uniform. To simplify notation, we will refer to $\cent(\alpha,s)$ as $\cent$.

\begin{algorithm}[H]
\caption{Approximate limited-attention PageRank$(V,E,s,\alpha,\delta)$}
\begin{algorithmic}[1]
\STATE $\epsilon = {\delta || s || _{1}}/{ |V| d^{max}_{out}}$;
\STATE $r = s$;
\STATE Queue q = new Queue();
\FOR{each $i \in V$}
\STATE $\tilde{\cent}[i] = 0$;
\IF{$\frac{r[i]}{ d^{max}_{out}} > \epsilon$}
\STATE q.add($i$);
\ENDIF
\ENDFOR
\WHILE{q.size() $>$ 0}
\STATE $i$ = q.dequeue();
\STATE $\tilde{\cent}[i] = \tilde{\cent}[i] + (1-\alpha) r[i]$;
\STATE $T = \alpha {r[i]}/{d_{out}(i)}$;
\STATE $r[i] = 0$;
\FOR{each $j \in N^{out}(i)$}
\STATE $r[j] = r[j] + {T }/{d_{in}(j)}$;
\IF{!q.contains($j$) and ${r[j]}/{ d^{max}_{out}} > \epsilon$}
\STATE q.add($j$);
\ENDIF
\ENDFOR
\ENDWHILE
\RETURN $\tilde{\cent}$;
\end{algorithmic}
\label{alg:lapr}
\end{algorithm}

\begin{theorem}
Given an $0 \le \alpha < 1$ and a uniform starting vector $s$, the approximate centrality vector $\tilde{\cent}$ is obtained from the algorithm in run time $O\big(\frac{|V| d^{out}_{max}}{(1-\alpha)\delta}\big)$ .
\label{theorem1}
\vspace{-0.41em}
\end{theorem}
\begin{proof}
Given an $\alpha$ in $[0,1]$. Algorithm~\ref{alg:lapr} works by dividing $\alpha r[i]$ equally amongst all $N^{out}(i)$ out-neighbors of node $i$. Each out-neighbor $j$ receives a fraction of the weight, based on its capacity, $d_{in}(j)$, to receive incoming messages. Hence, all $r[j]$ will increase by some fraction.
Let $r$ be old residual vector and $r^{\prime}$ be the updated residual vector. The sum of all elements of residual vector $\sum{r^{\prime}}$ is
\begin{equation*}
\sum{r^{\prime}}= \sum{r} - r[i] + \left( \alpha \frac{r[i]}{d_{out}(i)} \cdot \sum_{j \in N^{out}(i)}\frac {1} {d_{in}(j)} \right)
\end{equation*}
\noindent The sum of the entries of residual vector decreases by
\begin{eqnarray*}
\sum{r} & -&  \left( \sum{r} - r[i] + \frac{\alpha r[i]}{d_{out}(i)} \sum_{j \in N^{out}(i)} \frac{1}{d_{in}(j)} \right)  \\
\end{eqnarray*}
\begin{eqnarray*}
&= & r[i] - \left( \alpha \frac{r[i]}{d_{out}(i)} \cdot \sum_{j \in N^{out}(i)} \frac{1}{d_{in}(j)} \right)  \\
&>&  r[i] - \left( \alpha \frac{r[i]}{d_{out}(i)} \cdot d_{out}(i) \right)  >  (1 - \alpha) \epsilon d^{max}_{out} \\
\end{eqnarray*}
Let $k$ be the total number of iterations, net amount removed from residual vector will be at least
\begin{eqnarray*}
& & k(1 - \alpha) \epsilon d^{max}_{out} < ||s||_1
\implies   k <  \frac{||s||_1}{(1-\alpha)\epsilon d^{max}_{out}}
\end{eqnarray*}
Since each iteration is proportional to $d_{out}[i]$, the worst case time complexity is $O(\frac{||s||_1}{(1-\alpha)\epsilon})$. For our choice of $\epsilon$, this is equivalent to $O(\frac{|V|d^{max}_{out}}{\delta (1-\alpha)})$.
\end{proof}

\subsection{Approximate Limited-Attention Alpha-Centrality}
Limited attention Alpha-Centrality ($laAC$), given by Eq.~\ref{eq:lacr}, can be rewritten as the solution $\cent({\alpha,s})$ of:
\begin{equation*}
\cent(\alpha, s)[i]=s[i] +\alpha \sum_{j \in N^{out}(i)} \frac{\cent(\alpha,s)[j]}{ d_{in}(j)},
\end{equation*}
\noindent with the starting vector $s[i]=\sum_{j \in N^{out}(i)} {1}/{ d_{in}(j)}$. As before, we use $N^{out}(i)$  to denote the set of out-neighbors, and $N^{in}(i)$ the in-neighbors, of node $i$.

\begin{algorithm}[H]
\caption{Approximate Limited-Attention Alpha-Centrality($V,E,s,\alpha,\delta$)}
\begin{algorithmic}[1]
\STATE $\epsilon = \frac{\delta ||s|| _{1}}{ |V| d^{max}_{in}}$;
\STATE $r = s$;
\STATE Queue q = new Queue();
\FOR{each $i \in V$}
\STATE $\tilde{\cent}[i] = 0$;
\IF {$\frac{r[i]}{d^{max}_{in}}> \epsilon$}
\STATE q.add($i$);
\ENDIF
\ENDFOR
\WHILE{q.size() $>$ 0}
\STATE $i$ = q.dequeue();
\STATE $\tilde{\cent}[i] = \tilde{\cent}[i] +  r[i]$;
\STATE $T = \alpha \cdot \frac{r[i]}{d_{in}(i)}$;
\STATE $r(u) = 0$;
\FOR{each $j \in N^{in}(i)$}
\STATE $r[j] = r[j] + T $;
\IF{!q.contains($i$) and $\frac{r[j]}{d^{max}_{in}} > \epsilon$}
\STATE q.add($j$);
\ENDIF
\ENDFOR
\ENDWHILE
\RETURN $\tilde{\cent}$;
\end{algorithmic}
\label{alg:laac}
\end{algorithm}

\begin{theorem}
Given $0 < \alpha <1$ and starting vector $s$, the approximate centrality vector $\tilde{\cent}$ is obtained from the algorithm in run time $O(\frac{|V| d^{max}_{in}}{\delta (1-c)})$.
\label{theorem2}
\end{theorem}

\begin{proof}
Given an $\alpha$ in [0,1]. Let $r$ be old residual vector and $r^{\prime}$ be the updated residual vector.
The sum of all elements of residual vector $\sum{r^{\prime}}$ is
\begin{eqnarray*}
&\sum{r^{\prime}}= \sum{r} - r[i] + \left(\alpha d_{in}(j) \cdot \frac{r[j]}{d_{in}(j)}\right) \\
&= \sum{r} - r[j] + \left(\alpha r[j]\right)
\end{eqnarray*}
\noindent The sum of the entries of residual vector decreases by
\begin{eqnarray*}
\sum{r} &-& \left( \sum{r} - r[j] + \left(\alpha r[j] \right) \right)  \\
&= & r[j] - \left(\alpha r[j]\right)\\
&= & (1 - \alpha)r[j] > (1 - \alpha) \epsilon d^{max}_{in}\\
\end{eqnarray*}
\noindent Let $k$ be the total number of iterations, net amount removed from residual vector will be at least
\begin{eqnarray*}
& k (1 - \alpha) \epsilon d^{max}_{in} < ||s||_1 \implies k < \frac{||s||_1}{(1-\alpha)\epsilon d^{max}_{in}}
\end{eqnarray*}
\noindent Since each iteration is proportional to $d_{in}$, so the worst case time complexity is $O(\frac{||s||_1}{(1-\alpha)\epsilon})$. For our choice of $\epsilon$ this is equivalent to $O(\frac{|V| d^{max}_{in}}{\delta (1-\alpha)})$.
\end{proof}

\subsection{Performance of Approximate Algorithms}
For relatively small networks (up to thousands of nodes), we compared centrality scores calculated by the approximate algorithms  to those calculated by their exact versions.

\begin{figure*}[htb] %
   \centering
   \begin{tabular}{@{}c@{}c@{}}
   \multicolumn{2}{c}{USAir} \\
   \includegraphics[width=0.95\columnwidth]{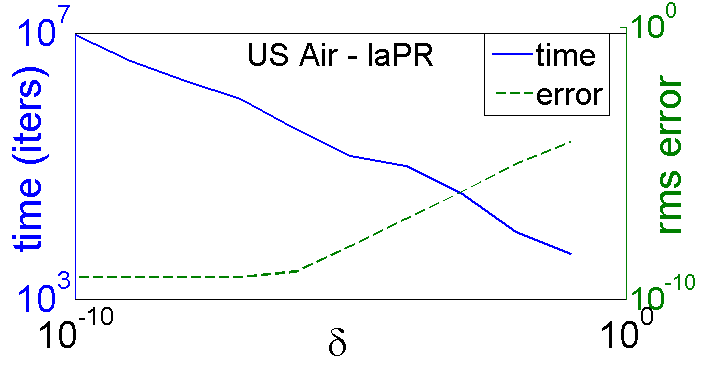} &
   \includegraphics[width=0.95\columnwidth]{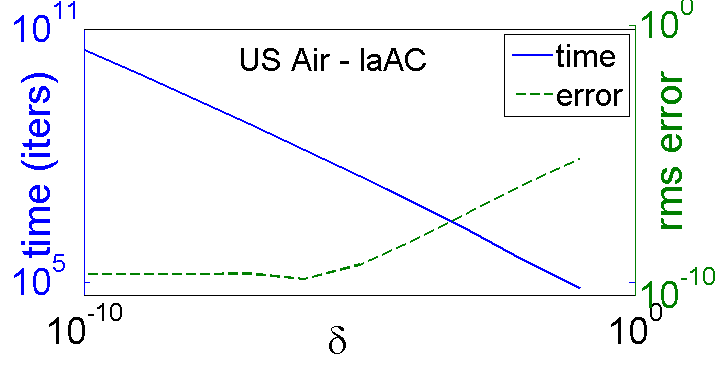} \\
   \multicolumn{2}{c}{Powergrid} \\
   \includegraphics[width=0.95\columnwidth]{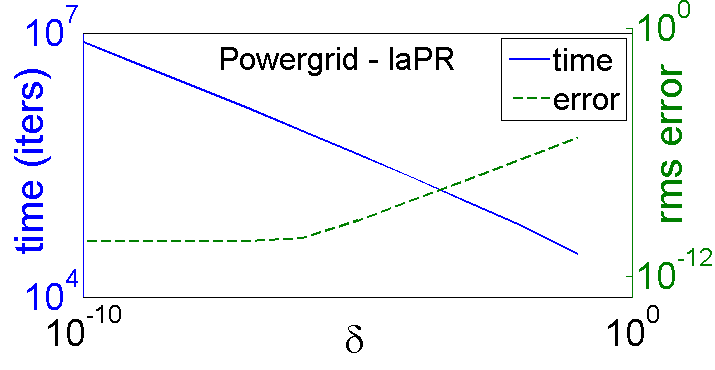} &
   \includegraphics[width=0.95\columnwidth]{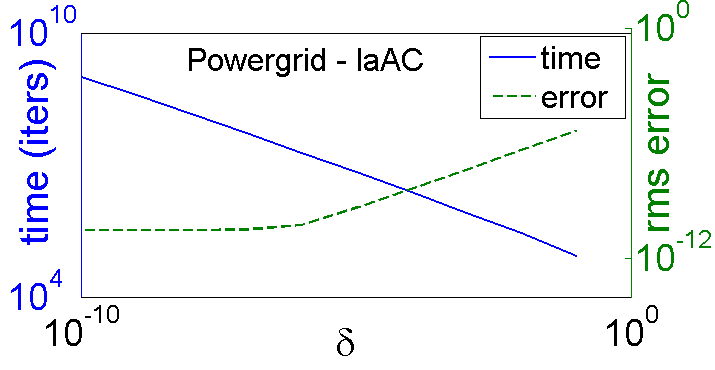} \\
   \multicolumn{2}{c}{Gnutella} \\
   \includegraphics[width=0.95\columnwidth]{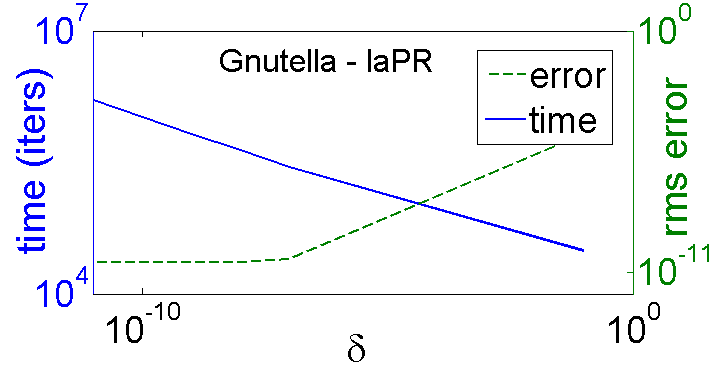} &
   \includegraphics[width=0.95\columnwidth]{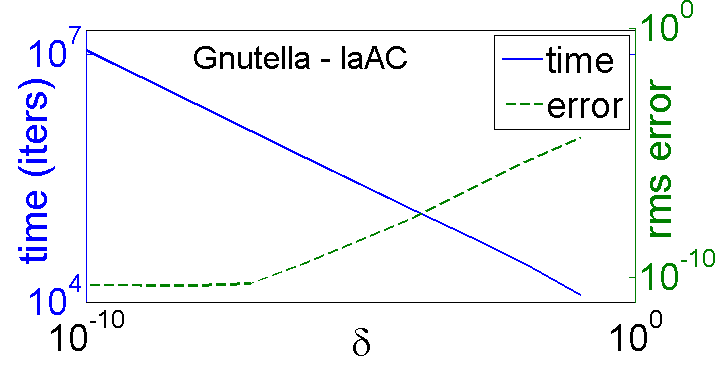} \\
   $laPR$ & $laAC$
   \end{tabular}
   \caption{Performance of the fast approximate limited-attention PageRank ($laPR$) and Alpha-Centrality ($laAC$) on Gnutella, US Air and Power grid networks. Performance is measured by time (number of iterations of the approximate algorithm) and $rms$ error of the centrality values calculated by the approximate and exact algorithms.}
   \label{fig:gnutella}
\vspace{-0.88em}
\end{figure*}

The $USAir$ network\footnote{http://vlado.fmf.uni-lj.si/pub/networks/data/} is  an undirected network of 332 nodes and 4,252 edges, which represent airports linked by direct flights. 
The $Powergrid$ network\footnote{http://cdg.columbia.edu/cdg/datasets} is an undirected network of 4,941 nodes and 6,594 edges representing the topology of the US Western States power grid.
The Gnutella dataset\footnote{http://snap.stanford.edu/data/} contains a snapshot of the Gnutella peer to peer network with 6,301 nodes and 20,777 edges.

Figure~\ref{fig:gnutella} shows the performance of the fast approximate algorithms proposed in this paper on the three networks vs the error tolerance $\delta$. Performance is measured in terms of time (number of iterations) taken to compute approximate centrality values and $rms$ error of these compared to the values computed by the exact algorithms Eqs.~\ref{eq:lapr} and \ref{eq:laac}. In all cases, while it takes longer to compute centrality scores for decreasing values of $\delta$, the answers are closer to their exact values.

\begin{figure}[htbp] %
   \centering
   \begin{tabular}{@{}c@{}}
   \includegraphics[width=0.5\textwidth]{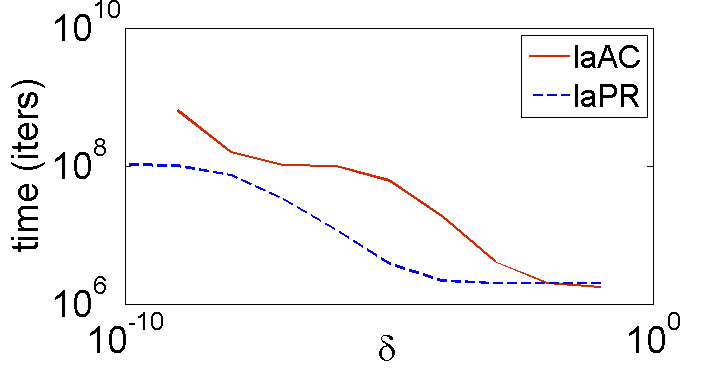} \\
   (a) Digg \\
   \includegraphics[width=0.5\textwidth]{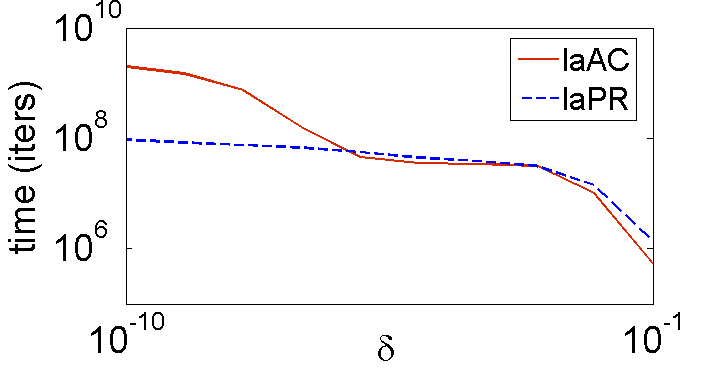} \\
    (b) Twitter
   \end{tabular}
   \caption{Number of iterations required to compute limited attention centralities for different value of of $\delta$ for (a) Digg and (b) Twitter networks. }
   \label{fig:iters}
\vspace{-0.88em}
\end{figure}
Figure~\ref{fig:iters} plots the number of iterations taken by the proposed algorithms to calculate  centralities for the Digg and Twitter data sets for different values of the error tolerance parameter $\delta$. Parameter values used in the calculations were $\alpha=9.0 \times 10^{-4}$ for both  $laAC$ and $laPR$ on Digg, and $\alpha=1 \times 10^{-4}$ for $laAC$ and $\alpha=0.9$ for $laPR$ on Twitter. As expected, the number of iterations increases for smaller error tolerances.

\bibliographystyle{plain}
\bibliography{references}


%

\end{document}